\documentclass[journal]{IEEEtran}

\usepackage{amsmath,amstext,amsfonts,amssymb,eucal,graphicx}
\usepackage{subfigure}
\usepackage{epsfig}
\usepackage{listings}

\newtheorem{lemma}{Lemma}

\newtheorem{proposition}{Proposition}
\newtheorem{theorem}{Theorem}
\newtheorem{definition}{Definition}

\newcommand{\E}{\mathbb{E}}

\begin{document}
\bibliographystyle{IEEEtran}%IEEEtran
\title{On the optimal stacking of noisy observations}
\author{
  \IEEEauthorblockN{\O yvind Ryan,~\IEEEmembership{Member,~IEEE}\\}
  \thanks{This work was supported by Alcatel-Lucent within the Alcatel-Lucent Chair on flexible radio at SUPELEC}
  \thanks{\O yvind~Ryan is with the Centre of Mathematics for Applications, University of Oslo, P.O. Box 1053 Blindern, NO-0316 Oslo, NORWAY, 
          and with SUPELEC, Gif-sur-Yvette, France, oyvindry@ifi.uio.no}
}

\markboth{Submitted to IEEE Transactions on Signal Processing}{Shell \MakeLowercase{\textit{et al.}}: Bare Demo of IEEEtran.cls for Journals}

\maketitle
\begin{abstract}
Observations where additive noise is present can for many models be grouped into a compound observation matrix, adhering to the same type of model. 
There are many ways the observations can be stacked, for instance vertically, horizontally, or quadratically. 
An estimator for the spectrum of the underlying model can be formulated for each stacking scenario in the case of Gaussian noise. 
We compare these spectrum estimators for the different stacking scenarios, 
and show that all kinds of stacking actually decreases the variance when 
compared to just taking an average of the observations. 
We show that, regardless of the number of observations, the variance of the estimator is smallest when the 
compound observation matrix is made as square as possible. 
When the number of observations grow, however, it is shown that the difference between the estimators is marginal: 
Two stacking scenarios where the number of columns and rows grow to infinity are shown to have the same variance asymptotically, even 
if the asymptotic matrix aspect ratios differ. 
Only the cases of vertical and horizontal stackings display different behaviour, giving a higher variance asymptotically.
Models where not all kinds of stackings are possible are also discussed.
\end{abstract}


\section{Introduction}
Random matrices find applications in many fields of research, 
such as digital communication~\cite{paper:telatar99}, mathematical finance~\cite{book:bouchaud} and nuclear physics~\cite{paper:guhr}. 
Free probability theory~\cite{vo2,paper:vomult,vo6,vo7,book:hiaipetz} has strong connections with random matrix theory, and can be used
for high dimensional statistical inference by addressing the following questions: 

{\em Given ${\bf A}$, ${\bf B}$ two $n\times n$ independent square Hermitian (or symmetric)
random matrices:\\
1) Can one derive the  eigenvalue distribution of ${\bf A}$ from the ones of ${\bf A} + {\bf B}$ and ${\bf B}$?\\
2) Can one derive the eigenvalue distribution of ${\bf A}$ from the ones
of ${\bf AB}$ and ${\bf B}$?}

More generally, such questions can be asked starting with any functional of the involved random matrices. 
If 1) or 2) can be answered for given random matrices ${\bf A}$ and ${\bf B}$, 
the corresponding operation for finding the eigenvalue distribution is called deconvolution. 
Deconvolution can be easier to perform in the large $n$-limit, and the literature contains result 
in this respect both for Vandermonde matrices~\cite{ryandebbah:vandermonde1,ryandebbah:vandermonde2}, 
and Gaussian matrices~\cite{eurecom:freedeconvinftheory}.
For Gaussian matrices, there also exist results in the finite regime~\cite{ryandebbah:finitedim}. 
The methods generally used to perform deconvolution are the Stieltjes transform method~\cite{paper:doziersilverstein1}, 
and the moments method~\cite{vo2,Florent}. 
In this contribution we will focus on the latter, which is based on the relations between the moments of the matrices involved. 
The $p$-th moment of an $n\times n$ random matrix ${\bf A}$ is defined as
\begin{equation} \label{momdef} 
t_{\bf A}^{n,p}=\E\left[ \mathrm{tr}({\bf A}^p) \right]=\int \lambda^pd\rho(\lambda)
\end{equation}
where $\E[\cdot]$ is the expectation, $\mathrm{tr}$ the normalized trace, 
and $d\rho=\E\left(\frac{1}{n} \sum_{i=1}^n \delta(\lambda-\lambda_i)\right)$ the associated empirical mean measure, 
with $\lambda_i$ the eigenvalues of ${\bf A}$. 
Both the Stieltjes transform and the moments can be used to retrieve eigenvalues, 
and can therefore be used for spectrum estimation.
For many types of random matrices, $t_{{\bf A}}^{n,p}$ converges almost surely when $n\to \infty$ to an analytical expression $t_{{\bf A}}^{p}$, 
depending only on some specific parameters, such as the distribution of the entries of ${\bf A}$.
This enables to reduce the dimensionality of the problem, and simplifies the computation of convolution of measures. 

Expressions for deconvolution turn out to be quite simple if asymptotic freeness~\cite{book:hiaipetz} is assumed, 
which is the case for large Gaussian matrices.
However, freeness does not hold in the finite case. In this respect,~\cite{ryandebbah:finitedim} modifies the moment-based free probability framework 
so that it applies for Gaussian matrices in the finite regime. 
The goal of this contribution is to address a particular question on how this framework best can be adapted for spectrum estimation purposes.
The observations in some random matrix models allow for stacking into a larger compound observation matrix.
When this is possible, some questions arise which do not seem to have been covered in the literature:
\begin{enumerate}
  \item Can the compound observation matrix be put into the same finite dimensional inference framework?
  \item Is one stacking of the observations better than another, for the purpose of spectrum estimation?
\end{enumerate}
A popular way of combining observations is the so-called sample covariance matrix, 
which essentially results from stacking observations of a random vector horizontally into a compound matrix. 
Results in this paper will actually challenge this construction for certain random matrix models, 
showing that it is not always the best way of combining observations.
Another facet of stacking is that it can make asymptotic results more applicable, and eliminate the need for results in the finite-dimensional regime. 
This can be very nice, since asymptotic results can be simpler to obtain, and have a nicer form. 
However, we will also present a model where stacking can only be partially applied.
The framework~\cite{ryandebbah:finitedim} has also been applied in a situation where it is not clear how to apply stacking in any way~\cite{romain:sensing}. 
We will give a partial answer to the above questions in this paper, 
in the sense that we characterize the stacking of observations which is optimal in terms of the variance of the corresponding spectrum estimators, 
and we characterize what we gain in comparison with methods where observations are not stacked.

The paper is organized as follows. 
Section~\ref{section:essentials} provides background essentials on random matrix theory 
needed to state the main results. In particular, we define different ways of stacking observations, and
random matrix models we will analyze which allow for such stacking. 
Section~\ref{section:theorems} states the main result, which concerns optimal stackings with the framework~\cite{ryandebbah:finitedim}. 
The main result can be stated without the actual definition of the spectrum estimator in question, 
which is therefore delayed till Section~\ref{mainestimator}. 
There we also prove formally that this estimator is unbiased, and state an expression for it, 
useful for implementing it and for proving the main result. 
In section~\ref{simulations} we present some useful simulations verifying the results.

\section{Random matrix Background Essentials} \label{section:essentials}
In the following, upper boldface symbols will be used for matrices,
and lower symbols will represent scalar values.
$(.)^T$ will denote the transpose operator, $(.)^\star$
conjugation, and $(.)^H=\left((.)^T\right)^\star$ hermitian
transpose. ${\bf I}_n$ will represent the $n\times n$ identity matrix.
We let $\mathrm{Tr}$ be the (non-normalized) trace for square matrices, defined by
\[
  \mathrm{Tr}({\bf A}) = \sum_{i=1}^n a_{ii},
\]
where $a_{ii}$ are the diagonal elements of the $n\times n$ matrix ${\bf A}$.
We also let $\mathrm{tr}$ be the normalized trace, defined by $\mathrm{tr}({\bf A}) = \frac{1}{n}\mathrm{Tr}({\bf A})$.
When ${\bf A}$ is non-random, we define its moments by $A_p=\mathrm{tr}({\bf A}^p)$, and more generally
\[ A_{p_1,\ldots,p_k} = \mathrm{tr}\left( {\bf A}^{p_1}\right) \cdots \mathrm{tr}\left({\bf A}^{p_k}\right).\]
If ${\bf A}$ is instead random, its (expected) moments are given by (\ref{momdef}), and more generally 
\[ \E\left[ \mathrm{tr}\left( {\bf A}^{p_1}\right) \cdots \mathrm{tr}\left({\bf A}^{p_k}\right) \right]. \]

${\bf X}$ will denote a standard complex Gaussian matrix, 
meaning that it has i.i.d. complex Gaussian entries with zero mean and unit variance
In particular, the real and imaginary parts of the entries are independent, each with variance $\frac{1}{2}$.

From $L=L_1L_2$ observations of an $n\times N$ random matrix ${\bf Y}$, we can form the $(nL_1)\times(NL_2)$ compound observation matrix, denoted ${\bf Y}_{L_1,L_2}$, 
by stacking the observations into a $L_1\times L_2$ block matrix in a given order.
Similarly, if ${\bf D}$ is non-random, we will denote by ${\bf D}_{L_1,L_2}$ the compound matrix formed in the same way from ${\bf D}$.
We will be concerned with the following question:

{\em 
Given a random matrix model on the form ${\bf Y}=f({\bf D},{\bf X}_1,{\bf X}_2,...)$, 
where ${\bf D}$ is non-random, 
and the ${\bf X}_i$ are Gaussian and independent. 
How can we best infer on the spectrum of ${\bf D}$ from independent observations ${\bf Y}_1,...,{\bf Y}_L$ of ${\bf Y}$?
}

We will restrict such inference to models where our methods apply directly to the compound observation matrix.
This turns out to be easier when one of the ${\bf X}_i$ is an additive component in $f$. 
To examplify this, we will first state the two models we will analyze.
For both, spectrum estimation methods from~\cite{ryandebbah:finitedim} will be applied.

\subsection{The additive model}
In applications such as MIMO channel modeling, the additive model
\begin{equation} \label{additivemodel}
  {\bf Y}=f({\bf D},{\bf X})={\bf D} + {\bf X},
\end{equation}
applies, where ${\bf D}$ is non-random and ${\bf X}$ is Gaussian, both $n\times N$. 
Given $L=L_1L_2$ observations of this model, the compound matrices ${\bf Y}_{L_1,L_2},{\bf D}_{L_1,L_2},{\bf X}_{L_1,L_2}$ satisfy
\[
  {\bf Y}_{L_1,L_2}=f({\bf D}_{L_1,L_2},{\bf X}_{L_1,L_2}).
\]
For (\ref{additivemodel}), we have moment-based methods to infer on the spectrum of $\frac{1}{N}{\bf D}{\bf D}^H$ 
from that of $\frac{1}{N}{\bf Y}{\bf Y}^H$~\cite{ryandebbah:finitedim,eurecom:channelcapacity}. 
Since ${\bf X}_{L_1,L_2}$ also is Gaussian, the same methods can be used to infer on the spectrum of $\frac{1}{NL_2}{\bf D}_{L_1,L_2}{\bf D}_{L_1,L_2}^H$ 
from the compound observation matrix. But since 
\begin{equation} \label{scalingeq}
  \mathrm{tr}\left(\left(\frac{1}{NL_2}{\bf D}_{L_1,L_2}{\bf D}_{L_1,L_2}^H\right)^p\right)
  = L_1^{p-1}\mathrm{tr}\left(\left(\frac{1}{N}{\bf D}{\bf D}^H\right)^p\right),
\end{equation}
spectrum estimation methods applied to the compound observation matrix actually helps us to infer on the spectrum of $\frac{1}{N}{\bf D}{\bf D}^H$.
We will state this estimator later on. To ease notation, we will let 
$D_p=\mathrm{tr}\left(\left( \frac{1}{N}{\bf D}{\bf D}^H\right)^p\right)$ in the following. 
We will see that different stackings $L_1,L_2$ give rise to different spectrum estimators, and compare their variances.

\subsection{A more involved model}
The random matrix model 
\begin{equation} \label{secondmodel}
  {\bf Y}=f({\bf D},{\bf X}_1,{\bf X}_2)={\bf D}{\bf X}_1 + {\bf X}_2
\end{equation}
can be found in multi-user MIMO applications. 
${\bf D}$ is non-random ($n\times m$), and ${\bf X}_1$ ($m\times N$) and ${\bf X}_2$ ($n\times N$) are independent and Gaussian. 
This model can also be subject to stacking, although the first component ${\bf D}{\bf X}_1$ in the sum now is random. 
To see this, assume that we have $L$ independent observations ${\bf Y}_k={\bf D}{\bf X}_{1,k} + {\bf X}_{2,k}$, $1\leq k\leq L$. 
Writing
\begin{eqnarray}
  {\bf Y}_{1...L}   &=& \left[{\bf Y}_{1},{\bf Y}_{2},...,{\bf Y}_{L}\right] \nonumber \\
  {\bf X}_{1,1...L} &=& \left[{\bf X}_{1,1},{\bf X}_{1,2},...,{\bf X}_{1,L}\right] \nonumber \\
  {\bf X}_{2,1...L} &=& \left[{\bf X}_{2,1},{\bf X}_{2,2},...,{\bf X}_{2,L}\right] \label{stackingeq}
\end{eqnarray}
(where ${\bf Y}_{1...L}$ is $n\times (NL)$, ${\bf X}_{1,1...L}$ is $m\times (NL)$, ${\bf X}_{2,1...L}$ is $n\times (NL)$), 
we can write ${\bf Y}_{1...L}={\bf D}{\bf X}_{1,1...L}+{\bf X}_{2,1...L}$. 
This has the same form as the original model, so that the same type of spectrum estimation methods can also be used for the compound observation matrix.
The underlying spectrum estimation method is now a two-stage process, where we in the first stage infer on the expected moments 
\begin{equation} \label{firststage}
  \E\left[ \mathrm{tr}\left(\left({\bf D}\left(\frac{1}{NL}{\bf X}_{1,1...L}{\bf X}_{1,1...L}^H\right){\bf D}^H\right)^p\right) \right],
\end{equation}
and use these in a second stage to infer on the moments of ${\bf D}{\bf D}^H$. 
This will be demonstrated further in Section~\ref{simulations}.

In light of the two models we have mentioned, we will differ between the following:
\begin{definition}
We will call a stacking of $L=L_1L_2$ observations into a $L_1\times L_2$ block matrix 
\begin{itemize}
  \item a horizontal stacking if $L_1=1$, 
  \item a vertical stacking if $L_2=1$, or
  \item a rectangular stacking if the limit $c=\lim \frac{L_1}{L_2}$ exists as
    the number of observations grow to infinity, with $0<c<\infty$.
\end{itemize}
These three types of stackings are also denoted by $H$, $V$, and $R$, respectively.
\end{definition}

For (\ref{secondmodel}), a horizontal stacking of observations lended itself, and seems to be the only natural way of stacking.
This is in contrast to (\ref{additivemodel}) where any rectangular stacking applied. 
For other models, neither rectangular, nor horizontal stacking way work~\cite{romain:sensing}.
We will not attempt to classify for which models the different types of stackings are possible, 
but rather to compare the different types of stackings whenever they apply.

In the literature, horizontal stacking has been applied to (\ref{additivemodel})~\cite{eurecom:channelcapacity}. 
The framework in~\cite{ryandebbah:finitedim} can be used to 
define unbiased estimators, useful for inference both for (\ref{additivemodel}) and (\ref{secondmodel}). 
For (\ref{additivemodel}), we will denote by $\widehat{D_p}$ such an estimator for $D_p$. 
Applying the framework to the compound observation matrices as sketched above, we get other unbiased estimators for $D_p$, 
denoted $\widehat{D_{p,L_1,L_2}}$. 
The main result, stated in the next section, is concerned with finding the "optimal" stacking,
in the sense of finding which $L_1,L_2$ give an estimator $\widehat{D_{p,L_1,L_2}}$ with lowest possible variance.
We will let $v_{p,\cdot,L}$ denote the variance of $\widehat{D_{p,L_1,L_2}}$, 
with $L=L_1L_2$ the number of observations, and $\cdot$ the stacking ($H$, $V$, or $R$). 
We will in addition let $A$ denote taking the average of $L$ applications of $\widehat{D_p}$, 
and denote the variance of the corresponding estimator by $v_{p,A,L}$. 
All stackings will be compared with averaging also. 
The main result is stated before the formulations of the estimators themselves (see Section~\ref{mainestimator}), 
since these expressions are rather combinatorial in nature, and thus need some more preliminaries before they can be stated.

\section{Statement of the main result} \label{section:theorems}
The following result essentially says that any rectangular stacking is asymptotically the best one when our framework is used, 
and that horizontal and vertical stacking have a slightly higher variance. 
Averaging of observations gives a variance higher than this again.
An even stronger conclusion can be drawn in the case of any given finite number of observations, 
where we prove that the stacking with "the most square" compound observation matrix gives the lowest variance.
This challenges the classical way of stacking observations horizontally when forming the sample covariance matrix. 

When $P$ is a polynomial in several variables, 
denote by the degree of $P$, or $\deg(P)$, the highest sum of exponents in any term therein. 
\begin{theorem} \label{stackable}
All estimators $\widehat{D_p},\widehat{D_{p,L_1,L_2}}$ we consider are unbiased estimators for $D_p=\mathrm{tr}\left(\left(\frac{1}{N}{\bf D}{\bf D}^H\right)^p\right)$ 
from observations of (\ref{additivemodel}), 
and their variances $v_{p,\cdot,L}$ are all $O(L^{-1})$. Moreover, 
\[ \lim_{L\rightarrow\infty} Lv_{1,\cdot,L} = \frac{2}{nN} D_1 + \frac{1}{nN},\]
where $\cdot$ can be $H,V,R$, or $A$. For $p\geq 2$ we have that
\begin{eqnarray*}
  \lim_{L\rightarrow\infty} Lv_{p,R,L} &=& \frac{2p^2}{nN} D_{2p-1} \\
  \lim_{L\rightarrow\infty} Lv_{p,V,L} &=& \frac{2p^2}{nN} D_{2p-1} + \frac{p^2}{N^2} D_{2p-2} \\
  \lim_{L\rightarrow\infty} Lv_{p,H,L} &=& \frac{2p^2}{nN} D_{2p-1} + \frac{p^2}{nN}  D_{2p-2} \\
  \lim_{L\rightarrow\infty} Lv_{p,A,L} &=& \frac{2p^2}{nN} D_{2p-1} + \left( \frac{p^2}{N^2} + \frac{p^2}{nN} \right) D_{2p-2} \\
                                       & & + Q(D_{2p-3},...,D_1),
\end{eqnarray*}
where $Q$ is a polynomial in $D_{2p-3},D_{2p-4},\ldots,D_1$ of degree $2p-2$, with only positive coefficients.
In particular, all rectangular stackings asymptotically have the same variance, and
\begin{eqnarray*}
  \lim_{L\rightarrow\infty} Lv_{p,R,L} &\leq& \lim_{L\rightarrow\infty} Lv_{p,V,L} \leq \lim_{L\rightarrow\infty} Lv_{p,A,L} \\
  \lim_{L\rightarrow\infty} Lv_{p,R,L} &\leq& \lim_{L\rightarrow\infty} Lv_{p,H,L} \leq \lim_{L\rightarrow\infty} Lv_{p,A,L}
\end{eqnarray*}
(since $Q(D_{2p-3},...,D_1)\geq 0$, since all $D_p\geq 0$). 
Also, the variance decreases with $L$ for a fixed stacking aspect ratio, and, 
for a given $L$ and any rectangular stackings $R_1,R_2$ into $L=L_1^{(1)}\times L_2^{(1)}$ and $L=L_1^{(2)}\times L_2^{(2)}$ observations, respectively. 
$v_{p,R_1,L}<v_{p,R_2,L}$ if and only if the $(nL_1^{(1)})\times(NL_2^{(1)})$ compound observation matrix 
is more square than the $(nL_1^{(2)})\times(NL_2^{(2)})$ compound observation matrix.
Also, $v_{p,\cdot,L}<v_{p,A,L}$ for any stacking.
\end{theorem}

The proof of Theorem~\ref{stackable} can be found in Appendix~\ref{mainapp}. 
The polynomial $Q$ above can be computed for the lower order moments, to compare the actual difference between averaging and stacking. 
Although we do not state the expression for $Q$, we have computed its values in an implementation in Section~\ref{simulations}, 
in order to show the actual variances for the different stacking scenarios. 

Since Theorem~\ref{stackable} is a statement on the leading order term of the variances of moments of certain random matrices, it is in the same genre as 
the recently developed theory of second order freeness~\cite{secondorderfreeness1,secondorderfreeness2,secondorderfreeness3}.
Our matrix setting is, however, slightly different than the ones considered in these papers. 

We will not state expressions for the variance for the model (\ref{secondmodel}), since this is more involved. 
Instead, we will in Section~\ref{simulations} verify in a simulation that stacking seems to be desirable here as well.
In the next section we will formulate the unbiased estimators for our models, which the main result refers to.

\section{Formulation of the estimator} \label{mainestimator}
To state our estimators, we need the following concept, taken from~\cite{ryandebbah:finitedim}:
\begin{definition}
  Let $p$ be a positive integer. By a partial permutation we mean a one-to-one mapping $\pi$
  between two subsets $\rho_1,\rho_2$ (which may be empty) of $\{ 1,\ldots,p\}$. We denote by $|\rho_1|$ the number of elements in $\rho_1$, 
  and by $\text{SP}_p$ the set of partial permutations of $p$ elements.
\end{definition}

$\pi$ is uniquely defined from the sets $\rho_1,\rho_2$, and a one-to one mapping $q:\rho_1\rightarrow\rho_2$. 
We will therefore in the following denote a partial permutation by $\pi=\pi(\rho_1,\rho_2,q)$. 
We will need the following result, taken from~\cite{ryandebbah:finitedim}, 
where the general statement is for the case when ${\bf D}$ is random, independent from ${\bf X}$:
\begin{proposition} \label{recursivesum}
Let ${\bf X}$ be an $n\times N$ standard, complex, Gaussian matrix and ${\bf D}$ be an $n\times N$ non-random matrix. Set
\begin{align*}
  D_{p_1,\ldots,p_k} &= \mathrm{tr}\left(\left(\frac{1}{N}{\bf D}{\bf D}^H\right)^{p_1} \right) \mathrm{tr}\left(\left(\frac{1}{N}{\bf D}{\bf D}^H\right)^{p_2} \right) \cdots \\
                     &  \qquad \times \mathrm{tr}\left(\left(\frac{1}{N}{\bf D}{\bf D}^H\right)^{p_k} \right) \\
  M_{p_1,\ldots,p_k} &= \E\left[ \mathrm{tr}\left(\left( \frac{1}{N}({\bf D}+{\bf X})({\bf D}+{\bf X})^H \right)^{p_1}\right) \right.\\
                     & \qquad \times \mathrm{tr}\left(\left( \frac{1}{N}({\bf D}+{\bf X})({\bf D}+{\bf X})^H \right)^{p_2}\right) \cdots \\
                     & \qquad \times \left. \mathrm{tr}\left(\left( \frac{1}{N}({\bf D}+{\bf X})({\bf D}+{\bf X})^H \right)^{p_k}\right)\right],
\end{align*}
We have that
\begin{align}\label{genformulasum}
  M_{p_1,...,p_k} &= \sum_{{\pi\in \text{SP}_p}\atop{\pi=\pi(\rho_1,\rho_2,q)}} \frac{n^{|\sigma(\pi)|-k}}{N^{|\rho_1|}} \nonumber \\
                  &\qquad\qquad \times N^{k(\rho(\pi))-kd(\rho(\pi))} n^{l(\rho(\pi))-ld(\rho(\pi))}  \nonumber \\
                  &\qquad\qquad \times D_{l_1,\ldots,l_r},
\end{align}
where $\sigma(\pi)$, $\rho(\pi)$, $k(\rho(\pi))$, $kd(\rho(\pi))$, $l(\rho(\pi))$, $ld(\rho(\pi))$, 
are explained below, and where $l_1,\ldots,l_r$ are the cardinalities of the blocks of $\sigma(\pi)$ divided by $2$. 
\end{proposition}

The following geometric interpretations explain the concepts in Proposition~\ref{recursivesum}, 
and is a summary from~\cite{ryandebbah:finitedim}:
\begin{itemize}
  \item We draw $k$ disconnected circles with $2p_1,2p_2,...,2p_k$ edges, respectively, and number the edges clockwise from $1$ to $2p_1+\cdots+2p_k$. 
    The set $\rho_1$ is visualized as a subset of the even edges ($2,4,...,2p$) under the mapping $i\rightarrow 2i$,
    $\rho_2$ is visualized as a subset of the odd edges ($1,3,...,2p-1$) under the mapping $i\rightarrow 2i-1$. 
  \item $q(i)=j$ means that the corresponding even and odd edges $2i$ and $2j-1$ are identified, 
    and with opposite orientation. 
  \item The vertices on the circles are also labeled clockwise, so that edge $i$ borders to vertices $i$ and $i+1$.
    When edges are identified as above, we also get an identification between the vertices bordering to the edges.
    This gives rise to an equivalence relation on the vertices. $\rho(\pi)\in{\cal P}(2p_1+\cdots+2p_k)$ 
    is the corresponding partition of the equivalence classes of vertices, where ${\cal P}(n)$ denotes the partition of $n$ elements. 
  \item It turns out that a block of $\rho$ either consists of odd numbers only (odd vertices), or of even numbers only (even vertices).
    $k(\rho(\pi))$ is defined as the number of blocks consisting of even numbers only, 
    $l(\rho(\pi))$ as the number of blocks consisting of odd numbers only.
  \item Edges from $\rho_1$ and $\rho_2$ are called random edges, other edges are called deterministic edges. 
    $kd(\rho(\pi))$ is the number of even equivalence classes of vertices bordering to a deterministic edge, 
    $ld(\rho(\pi))$ is defined similarly for odd equivalence classes of vertices.
  \item $\sigma=\sigma(\pi)$ is the partition where the blocks are the connected components of deterministic edges after identification of edges.
  \item By the graph of random edges we will mean the graph constructed when we, after the identification of edges, join vertices 
    which are connected with a path of deterministic edges, and afterwards remove the set of deterministic edges,
\end{itemize}
The quantities $k(\rho(\pi))-kd(\rho(\pi))$ and $l(\rho(\pi))-ld(\rho(\pi))$ in (\ref{genformulasum}) thus describe the number of even and odd vertices, respectively,  
which do not border to deterministic edges in the graph after the identification of edges. 
Note that when $\rho_1=\rho_2=\{1,...,p\}$, $\sigma(\pi)$ is a partition of zero elements. 
In this case we define $D_{l_1,\ldots,l_r}=1$. 

We now have all terminology in place in order to state a useful expression for our estimators:
\begin{lemma} \label{lemmaestimatorform}
Let ${\bf Y}={\bf D}+{\bf X}$ be an observation of the model (\ref{additivemodel}), and let $Y_p$ be the moments 
\begin{equation} \label{obsexpr}
Y_p=\mathrm{tr}\left(\left( \frac{1}{N}{\bf Y}{\bf Y}^H \right)^p\right). 
\end{equation}
\begin{align}
  \widehat{D_{p_1,...,p_k}} &= \sum_{{\pi\in \text{SP}_p}\atop{\pi=\pi(\rho_1,\rho_2,q)}} (-1)^{|\rho_1|} \frac{n^{|\sigma(\pi)|-k}}{N^{|\rho_1|}} \nonumber \\
                            &\qquad\qquad \times    N^{k(\rho(\pi))-kd(\rho(\pi))} n^{l(\rho(\pi))-ld(\rho(\pi))} \nonumber \\
                            &\qquad\qquad \times Y_{l_1,\ldots,l_r} \label{combform}
\end{align}
is an unbiased estimator for $D_{p_1,...,p_k}$, i.e. $\E\left(\widehat{D_{p_1,...,p_k}}\right)=D_{p_1,...,p_k}$ for all $p$. 
In particular, $\widehat{D_{p}}$ is an unbiased estimator for $D_p$. 

Similarly, given $L=L_1L_2$ observations of (\ref{additivemodel}), form the compound observation matrix ${\bf Y}_{L_1,L_2}$
and let instead $Y_p$ be the moments 
\begin{equation} \label{obsexpr2}
Y_p=\mathrm{tr}\left(\left( \frac{1}{NL_2}{\bf Y}_{L_1,L_2}{\bf Y}_{L_1,L_2}^H \right)^p\right).
\end{equation}
\begin{eqnarray}
  \lefteqn{\widehat{D_{p_1,...,p_k,L_1,L_2}}} \nonumber \\
  &=& L_1^{k-p_1-\cdots-p_k} \sum_{{\pi\in \text{SP}_p}\atop{\pi=\pi(\rho_1,\rho_2,q)}} (-1)^{|\rho_1|} \frac{(nL_1)^{|\sigma(\pi)|-k}}{(NL_2)^{|\rho_1|}} \nonumber \\
  & &\qquad\qquad \times    (L_2N)^{k(\rho(\pi))-kd(\rho(\pi))} (L_1n)^{l(\rho(\pi))-ld(\rho(\pi))} \nonumber \\
  & &\qquad\qquad \times Y_{l_1,\ldots,l_r} \label{combform2}
\end{eqnarray}
is also an unbiased estimator for $D_{p_1,...,p_k}$ for any $L_1,L_2$. 
In particular $\widehat{D_{p,L_1,L_2}}$ is an unbiased estimator for $D_p$.
\end{lemma}

When we talk about averaging of observations ${\bf Y}_1,...,{\bf Y}_L$, (i.e. the $A$ in $v_{p,A,L}$), 
we mean computing $\frac{1}{L}\sum_{i=1}^L \widehat{D_p}({\bf Y_i})$ using (\ref{combform}). 
It is clear that this is also an unbiased estimator for $D_p$, with variance $v_{p,A,L}$ being $\frac{1}{L}$ times that of $\widehat{D_p}$, 
since observations are assumed independent.

Note that there is a constant term in $\widehat{D_{p_1,...,p_k}}$, coming from $\pi$ where $\rho_1=\rho_2=\{1,...,p\}$. 
The proof of Lemma~\ref{lemmaestimatorform} can be found in Appendix~\ref{appendixlemmaestimatorform}, 
and builds on Proposition~\ref{recursivesum}. 
The appendix concentrates on the proof of (\ref{combform}), since the proof of (\ref{combform2}) is immediate: 
the term trailing $L_1^{k-p_1-\cdots-p_k}$ in (\ref{combform2})
is an unbiased estimator for the moments $F_p=\mathrm{tr}\left(\left(\frac{1}{NL_2}{\bf D}_{L_1,L_2}{\bf D}_{L_1,L_2}^H\right)^p\right)$, once (\ref{combform}) is proved, 
so that the entire right hand side of (\ref{combform2}) is an unbiased estimator for 
\begin {eqnarray*}
  \lefteqn{L_1^{k-p_1-\cdots-p_k} F_{p_1,\ldots,p_k}}\\ 
  &=& L_1^{k-p_1-\cdots-p_k} F_{p_1} \cdots F_{p_k}\\
  &=&(L_1^{1-p_1}F_{p_1})\cdots (L_1^{1-p_k}F_{p_k})\\
  &=& D_{p_1}\cdots D_{p_k},
\end{eqnarray*}
where we have used (\ref{scalingeq}).  

There can also be a known noise variance $\sigma$ present, so that 
(\ref{additivemodel}) takes the form ${\bf Y}={\bf D}+\sigma{\bf X}$. (\ref{combform}) can in this case be modified to
\begin{align}
  \widehat{D_{p_1,...,p_k}} &= \sum_{{\pi\in \text{SP}_p}\atop{\pi=\pi(\rho_1,\rho_2,q)}} (-1)^{|\rho_1|} \sigma^{2|\rho_1|}\frac{n^{|\sigma(\pi)|-k}}{N^{|\rho_1|}} \nonumber \\
                            &\qquad\qquad \times    N^{k(\rho(\pi))-kd(\rho(\pi))} n^{l(\rho(\pi))-ld(\rho(\pi))} \nonumber \\
                            &\qquad\qquad \times Y_{l_1,\ldots,l_r}. \label{combform3}
\end{align}
The proof of this is omitted, since it follows the same lines.

For the additive model (\ref{additivemodel}), we will only be interested in the expressions for the estimators $\widehat{D_p}$. 
If we have a model where ${\bf D}$ is instead a random matrix ${\bf R}$, like (\ref{secondmodel}), 
one can formulate unbiased estimators $\widehat{R_{p_1,...,p_k}}$ 
in the same way following~\cite{ryandebbah:finitedim}, unbiased now meaning $\E\left(\widehat{R_{p_1,...,p_k}}\right)=R_{p_1,...,p_k}$, where 
\begin{align}
  R_{p_1,\ldots,p_k} &= \E\left[ \mathrm{tr}\left( \left(\frac{1}{N}{\bf R}{\bf R}^H\right)^{p_1} \right) \right. 
                                 \mathrm{tr}\left( \left(\frac{1}{N}{\bf R}{\bf R}^H\right)^{p_2} \right) \cdots \nonumber \\
                     & \qquad \left.\times \mathrm{tr}\left( \left(\frac{1}{N}{\bf R}{\bf R}^H\right)^{p_k} \right) \right]. \label{rdef}
\end{align}                  

The estimators (\ref{combform}) were also used in~\cite{eurecom:channelcapacity}, without mentioning the form (\ref{combform}). 
This form is useful in that it makes it clear that the expressions (\ref{genformulasum}) and (\ref{combform}) are quite similar, 
enabling reuse of the implementation developed in~\cite{ryandebbah:finitedim} for computing (\ref{genformulasum}).
Secondly, (\ref{combform}) can be used for obtaining an expression for the variances
of $\widehat{D_p}$, staying within the same framework of partitions. 
We will state this expression and prove it in Appendix~\ref{appendixlemmaestimatorform}. 
In Appendix~\ref{mainapp}, Theorem~\ref{stackable} will be proved by analyzing this expression for the different stackings.

\section{Simulations} \label{simulations}
For~\cite{ryandebbah:finitedim} an implementation of the concepts used in Proposition~\ref{recursivesum} was made. 
In the following simulations, the computation of (\ref{actualexpression}) and (\ref{obsexpr2}) has used this implementation, 
with the restriction to the particular class of partitions therein
\footnote{A guide to the Matlab source code running the following simulations can be found in~\cite{ryandebbah:optstackingtools}.}.

Figure~\ref{fig:estimator} shows results for the third moment estimator (\ref{combform2}) applied to a diagonal matrix ${\bf D}$, 
with diagonal entries assumed to be $2,1,1,0.5$ (i.e. $n=N=4$).
The estimator where applied to quadratic stackings of $L=1,4,9,16,...,$ all the way up to $L=900$ observations. 
\begin{figure}
  \begin{center}
    \epsfig{figure=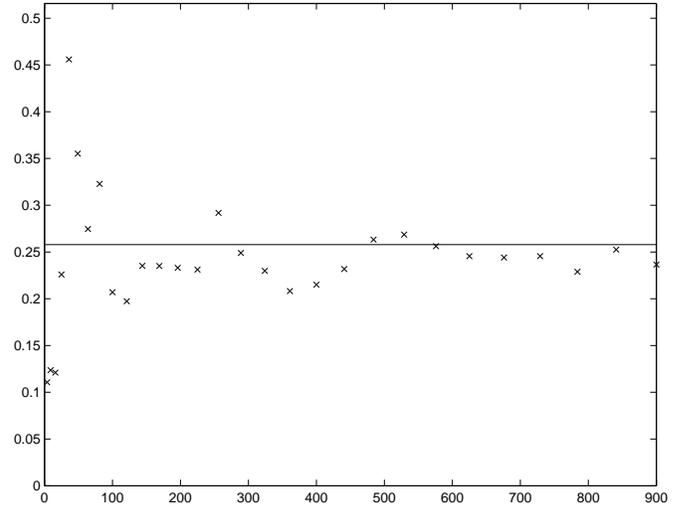,width=0.99\columnwidth}
  \end{center}
  \caption{The estimator (\ref{obsexpr2}) with quadratic stacking applied with different number of observations. 
           ${\bf D}$ is a $4\times 4$ matrix.
           The actual third moment of $\frac{1}{N}{\bf D}{\bf D}^H$ is also shown.} \label{fig:estimator}
\end{figure}
Although Theorem~\ref{stackable} says that the quadratic stacking is optimal, the difference between the different estimators may be hard to detect in practice, 
since differences may be small. 
Figure~\ref{fig:variancensadditive} gives a comparison for the actual variances for different number of observations 
and different stacking aspect ratios, verifying Theorem~\ref{stackable}. 
The theoretical limits for rectangular and horizontal stacking and averaging are also shown.
We have used the same $4\times 4$ matrix, and computed the expression (\ref{actualexpression}) 
to obtain the variance for the estimator for the third moment. 
As predicted by Theorem~\ref{stackable}, the variance tends towards the theoretical lower bounds for rectangular and horizontal stacking 
when the number of observations grow. 
For $L=50$ observations, to verify the results, we have also plotted the empirical variances
\[
  \frac{1}{K-1}\sum_{i=1}^K (x_i-\bar{x})^2,
\]
where $\{ x_i\}_{i=1}^K$ are $K$ outputs from the estimator (i.e. a number of $KL$ observations is needed, 
since each run of the estimator requires $L$ observations), 
and $\bar{x}=\frac{1}{K}\sum_{i=1}^K x_i$ is the mean.
We have set $K=1000$, and indicated the empirical variances for $L_1=1,2,5,10$, which correspond to $c=0.02,2/25,0.5,2$. 
\begin{figure}
  \subfigure[$L=5$]{\epsfig{figure=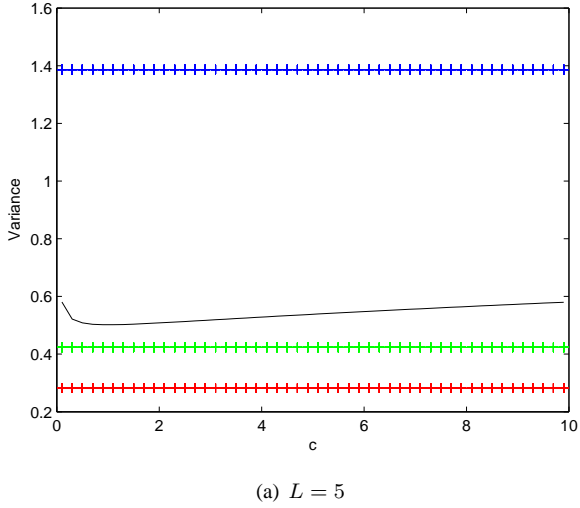,width=0.99\columnwidth}}
  \subfigure[$L=50$. Empirical variances are also shown for $c=0.02,2/25,0.5,2$.]{\epsfig{figure=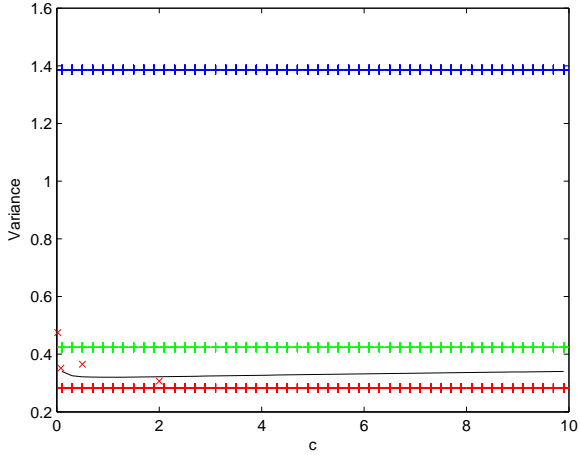,width=0.99\columnwidth}}
  \caption{Figures displaying $Lv_{3,\cdot,L}$ for the different estimators for the model (\ref{additivemodel}), for different number of observations $L$. 
    A diagonal matrix ${\bf D}$ with entries $2,1,1,0.5$ on the diagonal has been chosen.
    The three rectangular lines are the theoretical limits $\lim_{L\rightarrow\infty} Lv_{3,\cdot,L}$ 
    for rectangular stacking, horizontal stacking, and averaging, as predicted by Theorem~\ref{stackable}, in increasing order. 
    It is seen that aspect ratio near $1$ gives lowest variance, and that the variances decreases towards the theoretical limit predicted by 
    Theorem~\ref{stackable} when $L$ increases.} \label{fig:variancensadditive}
\end{figure}

\subsection{The model (\ref{secondmodel})}
We will compare horizontal stacking for (\ref{secondmodel}) with that of averaging.
As previously mentioned, this is a two-stage estimation, 
where we in the first stage get an unbiased estimate of the expected moments (\ref{firststage}) in the case of horizontal stacking, 
and an unbiased estimate of the expected moments $\E\left(\left({\bf D}\left(\frac{1}{N}{\bf X}_1{\bf X}_1^H\right){\bf D}^H\right)^p\right)$ 
in the case of averaging. In any case, denote the involved matrix by ${\bf S}$, define
\begin{equation}
  S_{p_1,...,p_k}=\E\left({\bf S}^{p_1}\right) \E\left({\bf S}^{p_2}\right) \cdots \E\left({\bf S}^{p_k}\right),
\end{equation}
and denote by $\widehat{S_{p_1,...,p_k}}$ the corresponding unbiased estimator. 
In the second stage, Theorem 3 of~\cite{ryandebbah:finitedim} gives unbiased estimators $\widehat{D_{p_1,...,p_k}}$ 
for the moments of ${\bf D}{\bf D}^H$ from the $\widehat{S_{p_1,...,p_k}}$, 
by stating an invertible matrix $A$ so that 
\[
  [\widehat{D_{p_1,...,p_k}}] = A^{-1}[\widehat{S_{p_1,...,p_k}}],
\]
where $[\widehat{S_{p_1,...,p_k}}]$ are all expected moments (for all possible $p_1,...,p_k$), grouped into a column vector in a given order.
In Figure~\ref{fig:model2est}, the unbiased estimators for horizontal stacking and averaging of observations have been compared for (\ref{secondmodel}). 
\begin{figure}
  \begin{center}
    \epsfig{figure=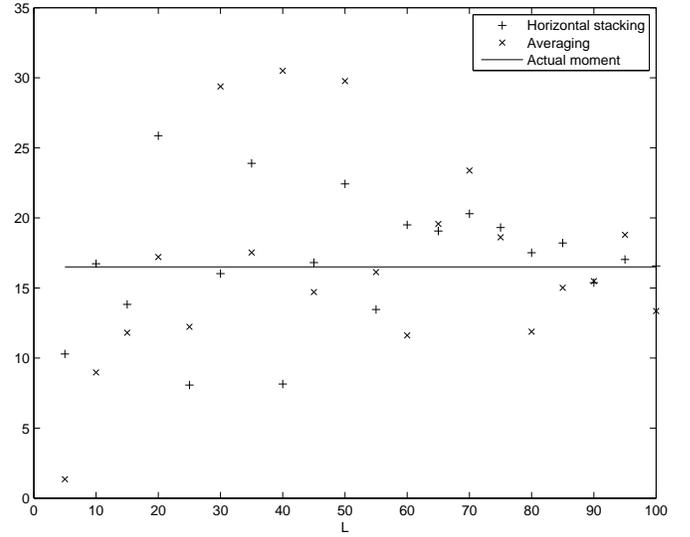,width=0.99\columnwidth}
  \end{center}
  \caption{The unbiased estimator for (\ref{secondmodel}) for the third moment of ${\bf D}{\bf D}^H$, 
           with ${\bf D}$ the $4\times 4$ diagonal matrix with $2,1,1,0.5$ on the diagonal. 
           The estimator is applied for up to $100$ observations, for both cases of horizontal stacking and averaging of observations.} \label{fig:model2est}
\end{figure}
The simulation is run for the same $4\times 4$ matrix, and it is seen that there is a high variance in the estimator for such a small matrix, 
even when the number of observations grows to $L=100$. To get an idea on whether horizontal stacking gives something here also in terms of variance, 
we need to run the estimators many times, and compare their empirical variances. 
This has been done in Figure~\ref{fig:model2emp}, where the empirical variance is computed from $50$ runs of the estimator. 
The figure suggests that, indeed, the empirical variance is lower in the case of stacking. 
We will, however, not prove this mathematically. 
\begin{figure}
  \begin{center}
    \epsfig{figure=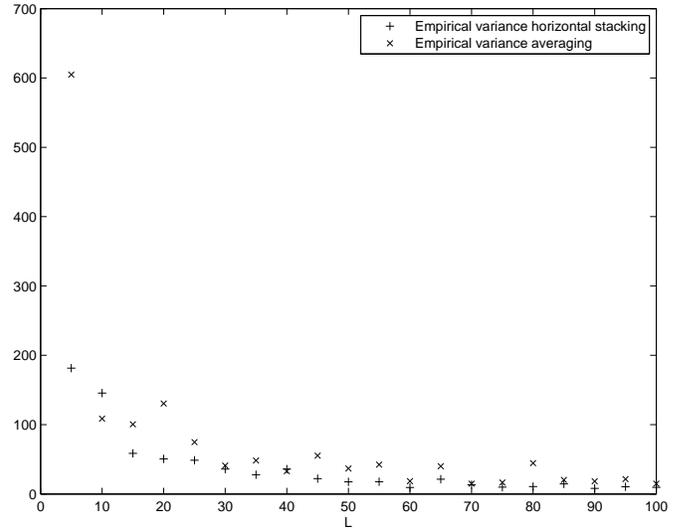,width=0.99\columnwidth}
  \end{center}
  \caption{The empirical variances of the estimators for (\ref{secondmodel}), which were shown in Figure~\ref{fig:model2est}. 
           For each $L$ the estimator was run $50$ times on  a set of $L$ observations, and the empirical variance was computed from this.
           It seems that the empirical variance is lower for the case of horizontal stacking, suggesting that results on stacking valid for 
           (\ref{additivemodel}) may have validity for more general models also.} \label{fig:model2emp}
\end{figure}

\section{Conclusion and further work}
We have analyzed an unbiased spectrum estimator for a model with additive Gaussian noise, 
and shown that the way the observations are stacked can play a role. 
More specifically, it is desirable to make the compound observation matrix as square as possible, 
as this will give rise to estimators with lowest possible variance. 
Asymptotically (i.e. when the number of observations grow to infinity), 
the variance of the estimators are the same, with only vertical and horizontal stacking and averaging displaying 
different asymptotic behaviour. All cases of stacking were shown to reduce the variance when compared to averaging. 
An estimator for the spectrum in a more general model was also applied, and a simulation suggested that stacking was desirable there as well. 

In this contribution, we arrived at a concrete "sum of terms"-expressions for the variance in our model, 
and the proof for the lower variance in stacked observation models boiled down to more terms vanishing when $L\to\infty$ in this expression, 
when stacking is considered. 
Once formulas for the variance for estimators in more general models are found, this "vanishing of terms" may be proved to be a much more general phenomenon, 
making the concept of stacking even more useful.
Future papers may contribute further along this line by putting the concept of stacking into a more general framework, applicable to more general models. 
In such a framework, (\ref{scalingeq}) (for how the moments of ${\bf D}$ are connected to those of the compound matrix ${\bf D}_{L_1,L_2}$) should
be replaced with more general methods, similarly to how (\ref{secondmodel}) was handled using results from~\cite{ryandebbah:finitedim}.

This paper only considers estimators which perform averaging or stacking of observations. 
Future work could consider non-linear ways of combining observations, and compare results on these with the results obtained here.
Theorem~\ref{stackable} should also have some significance when the noise is not Gaussian, 
since many random matrices with non-Gaussian, i.i.d. entries display the same asymptotic behaviour as Gaussian matrices. 
Future work could also consider this, and explore to which extent results generalize to the finite regime.

\appendices

\section{The proofs of Lemma~\ref{lemmaestimatorform} and Lemma~\ref{lemmavarianceform}} \label{appendixlemmaestimatorform}
To ease the expressions in the following, we will set 
\[ P_{\pi}(n,N)=N^{k(\rho(\pi))-kd(\rho(\pi))} n^{l(\rho(\pi))-ld(\rho(\pi))}. \]
To prove that the estimators $\widehat{D_{p_1,...,p_k}}$ in (\ref{combform}) are unbiased, we will first find alternative recursive expressions for them, 
and prove by induction that these are unbiased. 
Assume that we have found unbiased estimators $\widehat{D_{q_1,...,q_l}}$ building on (\ref{genformulasum}), whenever $q_1+\cdots+q_l<p_1+\cdots+p_k$. 
Define $\widehat{D_{p_1,...,p_k}}$ by reorganizing (\ref{genformulasum}) to
\begin{eqnarray}
  \widehat{D_{p_1,...,p_k}} &= Y_{p_1,...,p_k} - \sum_{p\geq 1}\sum_{{\pi\in \text{SP}_p}\atop{\pi=\pi(\rho_1,\rho_2,q)}} \frac{n^{|\sigma(\pi)|-k}}{N^{|\rho_1|}} \nonumber \\
                        &\qquad\qquad\times P_{\pi}(n,N) \widehat{D_{l_1,\ldots,l_r}}. \label{recursiveversion}
\end{eqnarray}
Here the term for the empty partial permutation has been separated from the other terms, 
and, by convention, $\widehat{D_{l_1,\ldots,l_r}}=1$ whenever $\pi=\pi(\rho_1,\rho_2,q)$ with $\rho_1=\rho_2=\{1,...,p\}$. 
Taking expectations on both sides in (\ref{recursiveversion}) we get 
\begin{eqnarray*}
  \lefteqn{\E(\widehat{D_{p_1,...,p_k}})} \\
                                &=& \E(Y_{p_1,...,p_k}) \\
                                & & - \sum_{p\geq 1}\sum_{{\pi\in \text{SP}_p}\atop{\pi=\pi(\rho_1,\rho_2,q)}} \frac{n^{|\sigma(\pi)|-k}}{N^{|\rho_1|}} P_{\pi}(n,N) \E(\widehat{D_{l_1,\ldots,l_r}}) \\
                                &=& D_{p_1,...,p_k} \\
                                & & + \sum_{p\geq 1}\sum_{{\pi\in \text{SP}_p}\atop{\pi=\pi(\rho_1,\rho_2,q)}} \frac{n^{|\sigma(\pi)|-k}}{N^{|\rho_1|}} P_{\pi}(n,N) D_{l_1,\ldots,l_r} \\                            
                                & & - \sum_{p\geq 1}\sum_{{\pi\in \text{SP}_p}\atop{\pi=\pi(\rho_1,\rho_2,q)}} \frac{n^{|\sigma(\pi)|-k}}{N^{|\rho_1|}} P_{\pi}(n,N) D_{l_1,\ldots,l_r} \\
                                &=& D_{p_1,...,p_k},
\end{eqnarray*}
where we have again used (\ref{genformulasum}).
This shows that $\widehat{D_{p_1,...,p_k}}$ also is unbiased. 
We will now show that this recursive definition of $\widehat{D_{p_1,...,p_k}}$ coincides with (\ref{combform}), which will complete the proof 
of Lemma~\ref{lemmaestimatorform}. 

Recursively replacing the $\widehat{D_{l_1,...,l_r}}$ in (\ref{recursiveversion}) until there are only terms on the form $Y_{p_1,...,p_k}$ left, 
we arrive at an expression on the form
\begin{equation} \label{replacethis2}
\begin{array}{l}
  \sum_l \sum_{\pi_1,...,\pi_l} (-1)^l \left(\prod_{i=1}^l \frac{n^{|\sigma(\pi_i)|-\sigma(\pi_{i-1})}}{N^{|\rho_{1i}|}} P_{\pi_i}(n,N)\right) Y_{l_1,...,l_r} \\
  =\sum_l \sum_{\pi_1,...,\pi_l} (-1)^l \frac{n^{|\sigma(\pi_l)|-k}}{N^{\sum_{i=1}^l|\rho_{1i}|}} \left(\prod_{i=1}^l P_{\pi_i}(n,N)\right) Y_{l_1,...,l_r},
\end{array}
\end{equation}
where $\pi_1,...,\pi_l$ are non-empty partial permutations, 
and where $l_1,...,l_r$ are the cardinalities of the blocks after the identification of edges from all $\pi_1,...,\pi_l$. 
We will call a $\pi=\pi_1,...,\pi_l$ a {\em nested partial permutation}, since it corresponds to a nested application of partial permutations. 
The factor $(-1)^l$ comes from $l$ applications of (\ref{recursiveversion}), where each application contributes a $-1$ from therein. 
Due to this alternating sign, many terms in (\ref{replacethis2}) will cancel.
The following class of permutations will be useful to see these cancellations:
\begin{definition}
Let $\Pi_{l,k}$ be the set of nested partial permutations on the form $\{\pi_1,...,\pi_l\}$, 
where $|\rho_{\pi_1}|+\cdots + |\rho_{\pi_l}|=k$. 
Also, when $\pi=\{\pi_1,...,\pi_l\}$ are nested partial permutations which do not contain any identifications involving edges $i$ or $j$, 
let $\Pi_{\pi,i,j}\subset \Pi_{l,k+2}\cup \Pi_{l+1,k+2}$ be the set of nested partial permutations which equals $\pi$, 
with the exception that the identification $(i,j)$ is added. 
\end{definition}

It is clear that any $\pi\in\Pi_{\pi,i,j}$ gives equal contribution in (\ref{replacethis2}) up to sign, since each such $\pi$ embraces the same edges, 
and the order of the identification of edges does not matter for the final graph. 
It is also clear that
\begin{eqnarray*}
  | \Pi_{\pi,i,j}\cap \Pi_{l,k+2}|   &=& l \\
  | \Pi_{\pi,i,j}\cap \Pi_{l+1,k+2}| &=& l+1,
\end{eqnarray*}
and that the contributions from the two sets $\Pi_{\pi,i,j}\cap \Pi_{l,k+2}$ and $\Pi_{\pi,i,j}\cap \Pi_{l+1,k+2}$ have opposite signs, 
since the sign for any $\pi\in\Pi_{l,k}$ is $(-1)^l$.  
Adding the contributions, we get that the total contribution from $\Pi_{\pi,i,j}$ equals that from just one nested partial permutation in $\Pi_{l+1,k+2}$ 
where we set $\pi_{l+1}=(i,j)$.
Summing over all $\pi$ and $l$ where $\pi=\{\pi_1,...,\pi_l\}$ does not contain any identifications involving $i$ or $j$, we get that 
the contribution from the set of $\pi$ which contain $(i,j)$ equals the sum over $\{\pi_1,...,\pi_{l-1},\pi_l=(i,j)\}$. 
In the same way we can sum over $\pi$ with $(i,j)$ replaced by all other edge possibilities, to arrive at the sum over all 
$\pi=\{\pi_1,...,\pi_l\}$, where all $|\rho_{\pi_i}|=1$, and where we need only sum over sets (i.e. the order of the elements does not matter). 
In other words, and since there are $l=|\rho_1|$ partial permutations nested in this way, we can replace (\ref{recursiveversion}) with 
\[
  \widehat{D_p} = \sum_{{\pi\in \text{SP}_p}\atop{\pi=\pi(\rho_1,\rho_2,q)}} (-1)^{|\rho_1|} \frac{n^{|\sigma(\pi)|-1}}{N^{|\rho_1|}} 
        P_{\pi}(n,N) Y_{l_1,\ldots,l_r}.
\]
This coincides with (\ref{combform}), and the proof of Lemma~\ref{lemmaestimatorform} is finished.

We will have use for the following lemma, which states an expression for the variance of $\widehat{D_p}$. 
We will only state it for the case of no stacking, and leave the variance for the estimators using stackings to Appendix~\ref{mainapp}.
\begin{lemma} \label{lemmavarianceform}
Let $\text{SPR}_{2p}$ be the set of partial permutations of $\{ 1,...,2p\}$ such that all identifications are from $\{ 1,...,p\}$ to $\{ p+1,...,2p\}$, or vice versa. 
The variance
\[ v_p=\E\left[ \widehat{D_p}^2\right]-\E\left[\widehat{D_p}\right]^2 \] 
of $\widehat{D_p}$ equals 
\begin{equation} \label{actualexpression}
  \sum_{ \pi\in\text{SPR}_{2p} } \frac{n^{|\sigma(\pi)|-2}}{N^{|\rho_1|}}  P_{\pi}(n,N) D_{l_1,\ldots,l_r}.
\end{equation}
\end{lemma}

\begin{proof}
Inserting (\ref{combform}) twice we get 
\begin{eqnarray}
  v_p &=& \E\left[\widehat{D_p}^2\right]-\E\left[\widehat{D_p}\right]^2 \nonumber \\
      &=& \sum_{{\pi_1\in\text{SP}_p}\atop{\pi_1=\pi(\rho^{(1)}_1,\rho^{(1)}_2,q)}} 
          \sum_{{\pi_2\in\text{SP}_p}\atop{\pi_2=\pi(\rho^{(2)}_1,\rho^{(2)}_2,q)}} \nonumber \\
      & & \qquad\qquad(-1)^{|\rho^{(1)}_1|} (-1)^{|\rho^{(2)}_1|} 
          \frac{n^{|\sigma(\pi_1)|-1}}{N^{|\rho^{(1)}_1|}} \frac{n^{|\sigma(\pi_2)|-1}}{N^{|\rho^{(2)}_1|}} \nonumber \\
      & & \qquad\qquad\times P_{\pi_1}(n,N) P_{\pi_2}(n,N) \nonumber \\
      & & \qquad\qquad\times (\E\left[ Y_{l_1^{(1)},\ldots,l_{r_1}^{(1)}} Y_{l_1^{(2)},\ldots,l_{r_2}^{(2)}} \right] \nonumber \\
      & & \qquad\qquad\qquad - \E\left[ Y_{l_1^{(1)},\ldots,l_{r_1}^{(1)}}\right] \E\left[ Y_{l_1^{(2)},\ldots,l_{r_2}^{(2)}} \right]), \label{conthere}
\end{eqnarray}
where  $l_1^{(1)},\ldots,l_{r_1}^{(1)}$ are the cardinalities of the blocks of $\sigma(\pi_1)$ divided by $2$,
$l_1^{(2)},\ldots,l_{r_2}^{(2)}$ those of $\sigma(\pi_2)$. Using (\ref{genformulasum}) we can write
\begin{eqnarray}
  \lefteqn{\E\left[ Y_{l_1^{(1)},\ldots,l_{r_1}^{(1)}} Y_{l_1^{(2)},\ldots,l_{r_2}^{(2)}} \right]} \nonumber\\
  &=& \sum_{{\pi\in \text{SP}_{2p-|\rho^{(1)}_1|-|\rho^{(2)}_1|}}\atop{\pi=\pi(\rho_1,\rho_2,q)}}
      \frac{n^{|\sigma(\pi)|-r_1-r_2}}{N^{|\rho_1|}} P_{\pi}(n,N) D_{l_1,\ldots,l_r} \label{conthere1},
\end{eqnarray}
where $l_1,\ldots,l_r$ are the cardinalities of $\sigma(\pi)$ divided by $2$, and
\begin{eqnarray}
  \lefteqn{\E\left[ Y_{l_1^{(1)},\ldots,l_{r_1}^{(1)}}\right] \E\left[ Y_{l_1^{(2)},\ldots,l_{r_2}^{(2)}} \right]} \nonumber\\
  &=& \sum_{{\pi^{(1)}\in \text{SP}_{p-|\rho^{(1)}_1|}}\atop{\pi^{(1)}=\pi(\rho_{11},\rho_{12},q)}}
      \sum_{{\pi^{(2)}\in \text{SP}_{p-|\rho^{(2)}_1|}}\atop{\pi^{(2)}=\pi(\rho_{21},\rho_{22},q)}} \nonumber\\
  & & \qquad\qquad \frac{n^{|\sigma(\pi^{(1)})|-r_1}}{N^{|\rho_{11}|}} \frac{n^{|\sigma(\pi^{(2)})|-r_2}}{N^{|\rho_{21}|}} \nonumber\\
  & & \qquad\qquad\times P_{\pi^{(1)}}(n,N) P_{\pi^{(2)}}(n,N) \nonumber \\
  & & \qquad\qquad\times D_{l_1^{(1)},\ldots,l_{r_1}^{(1)}} D_{l_1^{(2)},\ldots,l_{r_2}^{(2)}}, \label{conthere2}
\end{eqnarray}
where $l_1^{(1)},\ldots,l_{r_1}^{(1)}$ are the cardinalities of $\sigma(\pi^{(1)})$ divided by $2$, 
$l_1^{(2)},\ldots,l_{r_2}^{(2)}$ those of $\sigma(\pi^{(2)})$.
The powers of $n$ and $N$ in (\ref{conthere2}) can be written
\[ \frac{n^{|\sigma(\pi^{(1)})|+|\sigma(\pi^{(2)})|-r_1-r_2}}{N^{|\rho_{11}|+|\rho_{21}|}} P_{\pi^{(1)}}(n,N) P_{\pi^{(2)}}(n,N),\]
which are seen to match the powers of $n$ and $N$ in (\ref{conthere1}) when $\pi=\pi_1\times\pi_2$ 
does not contain any identification of edges from different expectations.
These terms thus cancel, and we are left with summing over $\pi$ 
containing identification of edges between the two expectations.

To see that we need only sum over $\pi$ containing only identification of edges from one expectation to another,  
note that a $\pi_1$ containing $(i,j)$ cancels the contribution from a $\pi$ containing $(i,j)$, 
since the former has an additional power of $-1$. The same can be said for $\pi_2$. 
The only terms not canceling therefore occur when $\pi_1$ and $\pi_2$ are empty, 
and $\pi$ only contains identifications between the two expectations. 
These correspond to $\text{SPR}_{2p}$ by definition.
All of them contribute with a positive sign, and all in all we get that $v_p$ equals
\[
  \sum_{ \pi\in\text{SPR}_{2p} } 
  \frac{n^{|\sigma(\pi)|-2}}{N^{|\rho_1|}}  P_{\pi}(n,N) D_{l_1,\ldots,l_r}
\]
(since $r_1=r_2=1$), which is what we had to show.
\end{proof}

\section{The proof of Theorem~\ref{stackable}} \label{mainapp}
The geometric interpretation of $\pi\in\text{SPR}_{2p}$ is as an identification among some of $4p$ edges, 
where even edges are only identified with odd edges and vice versa, 
and where there are only identifications between $\{ 1,...,2p\}$ and $\{ 2p+1,...,4p\}$, and vice versa.
It is clear that $\pi\in\text{SPR}_{2p}$ is invariant under cyclic shifts of the form $\pi\rightarrow s_{1k}s_{2l} \pi (s_{1k}s_{2l})^{-1}$, where
\begin{eqnarray*}
  s_{1k}(r) &=& \left\{ \begin{array}{l} r+k \mbox{ for } r\in\{ 1,...,2p\} \\ r \mbox{ for } r\in\{ 2p+1,...,4p\} \end{array} \right. \\
  s_{2l}(r) &=& \left\{ \begin{array}{l} r \mbox{ for } r\in\{ 1,...,2p\} \\ r+l \mbox{ for } r\in\{ 2p+1,...,4p\} \end{array} \right.
\end{eqnarray*}
(addition performed so that result stays within the same interval, either $[1,...,2p]$ or $[2p+1,...,4p]$) 
as long as $k$ and $l$ either are both odd, or both even, in order for the identification to remain between even and odd elements and vice versa. 
The equivalence class of $\pi\in\text{SPR}_{2p}$ under cyclic shifts is given by 
$\cup_{k,l} s_{1k}s_{2l} \pi (s_{1k}s_{2l})^{-1}$, where $k$ and $l$ either are both odd, or both even. 
We will denote by $\text{SPE}_{2p}$ the set of such equivalence classes, and denote by $\bar{\pi}\in\text{SPE}_{2p}$ 
the equivalence class of $\pi\in\text{SPR}_{2p}$.

From the geometric interpretation of $\pi$ it is clear that, when we instead of $\pi$ use $s_{1k}s_{2l} \pi (s_{1k}s_{2l})^{-1}$,
\begin{enumerate}
  \item $|\rho_1|$ and $|\rho_2|$ is the same for $\pi$ and $s_{1k}s_{2l} \pi (s_{1k}s_{2l})^{-1}$,
  \item $|\sigma(\pi)|=|\sigma(s_{1k}s_{2l} \pi (s_{1k}s_{2l})^{-1})|$. 
    The block cardinalities $l_1,...,l_r$ of  $\sigma(\pi)$ and $\sigma(s_{1k}s_{2l} \pi (s_{1k}s_{2l})^{-1})$ are also equal,
  \item when $k$ and $l$ are both even, $k,kd,l,ld$ are the same for $\rho(\pi)$ and $\rho(s_{1k}s_{2l} \pi (s_{1k}s_{2l})^{-1})$,
  \item when $k$ and $l$ are both odd, 
    \begin{eqnarray*}
      k(\rho(\pi)  &=& l(\rho(s_{1k}s_{2l} \pi (s_{1k}s_{2l})^{-1})) \\
      kd(\rho(\pi) &=& ld(\rho(s_{1k}s_{2l} \pi (s_{1k}s_{2l})^{-1})) \\
      l(\rho(\pi)  &=& k(\rho(s_{1k}s_{2l} \pi (s_{1k}s_{2l})^{-1})) \\
      ld(\rho(\pi) &=& kd(\rho(s_{1k}s_{2l} \pi (s_{1k}s_{2l})^{-1})).
    \end{eqnarray*}
\end{enumerate}
By definition of $P_{\pi}$, the last two statements say that 
\[ P_{s_{1k}s_{2l} \pi (s_{1k}s_{2l})^{-1}}(n,N) = P_{\pi}(n,N) \] 
when $k,l$ are both even, and
\[ P_{s_{1k}s_{2l} \pi (s_{1k}s_{2l})^{-1}}(n,N) = P_{\pi}(N,n) \]
when $k,l$ are both odd. 
Since there are equally many elements with $k,l$ odd and $k,l$ even under cyclic equivalence, we see that
\begin{equation}
  Q_{\bar{\pi}}(n,N)=\sum_{\pi_1\sim\pi} P_{\pi_1}(n,N)
\end{equation}
is a polynomial symmetric in $n$ and $N$, where $\sim$ denotes equivalence under cyclic shifts.
The first statements above say that the rest of the powers of $n$ and $N$ in (\ref{actualexpression}) are unchanged under cyclic equivalence. 
By summing over the cyclic equivalence classes in (\ref{actualexpression}), we see that it can be rewritten to
\begin{equation} \label{newexpr}
  v_p = \sum_{ \bar{\pi}\in\text{SPE}_{2p} } \frac{n^{|\sigma(\pi)|-2}}{N^{|\rho_1|}}  Q_{\bar{\pi}}(n,N) D_{l_1,\ldots,l_r},
\end{equation}
with $Q_{\bar{\pi}}$  symmetric in $n$ and $N$. 
Moreover, $Q_{\bar{\pi}}$ has the form $Q_{\bar{\pi}}(n,N)=an^kN^l+bn^lN^k$, 
where $a+b$ is the number of elements in the cyclic equivalence class of $\pi$. 

Since $\widehat{D_{p,L_1,L_2}}$ is $L_1^{1-p}$ times the estimator for the $p$-th moment $F_p$ of the compound matrix 
by the comments following the statement of Lemma~\ref{lemmaestimatorform},
the variance $v_{p,\cdot,L}$ of $\widehat{D_{p,L_1,L_2}}$ in (\ref{combform2}) is, 
after replacing $n$ with $nL_1$, and $N$ with $NL_2$ in (\ref{newexpr}),
\begin{eqnarray}
  v_{p,\cdot,L} &=& L_1^{2-2p} \sum_{ \pi\in\text{SPE}_{2p} } 
                    \frac{n^{|\sigma(\pi)|-2}L_1^{|\sigma(\pi)|-2}}{N^{|\rho_1|}L_2^{|\rho_1|}} 
                    Q_{\bar{\pi}}(nL_1,NL_2) \nonumber \\
                & & \times F_{l_1,\ldots,l_r} \nonumber \\
                &=& L_1^{2-2p} \sum_{ \pi\in\text{SPE}_{2p} } 
                    \frac{n^{|\sigma(\pi)|-2}L_1^{|\sigma(\pi)|-2}}{N^{|\rho_1|}L_2^{|\rho_1|}} 
                    Q_{\bar{\pi}}(nL_1,NL_2) \nonumber \\
                & & \times L_1^{2p-|\rho_1|-|\sigma(\pi)|} D_{l_1,\ldots,l_r} \nonumber \\
                &=& \sum_{ \pi\in\text{SPE}_{2p} } 
                    \frac{n^{|\sigma(\pi)|-2}}{N^{|\rho_1|}L^{|\rho_1|}}  
                    Q_{\bar{\pi}}(nL_1,NL_2) 
                    D_{l_1,\ldots,l_r}, \label{lookatthis}
\end{eqnarray} 
where we have used (\ref{scalingeq}), and set $L=L_1L_2$.

$\deg(Q_{\bar{\pi}})$ describes the number of vertices in the graph of random edges not bordering to deterministic edges. 
Each vertex is associated with a value $\leq L\max(n,N)$, so that $Q_{\bar{\pi}}$ has order at most $L$ to the power of 
the number of vertices not bordering to deterministic edges.
We will use this in the following, and consider the following possibilities:
\begin{enumerate}
  \item There are no deterministic edges: in this case, $p=|\rho_1|/2$.
    Since there are only crossidentifications between $\{1,...,2p\}$ and $\{2p+1,...,4p\}$ for $\pi\in\text{SPR}_{2p}$, 
    any vertex in $\{2p+1,...,4p\}$ is identified with a vertex in $\{1,...,2p\}$, 
    so that $\{1,...,2p\}$ contains representatives for all equivalence classes of vertices. 
    There are thus at most $p$ even equivalence classes, and at most $p$ odd equivalence classes. Thus
    \begin{eqnarray*}
      Q_{\bar{\pi}}(nL_1,NL_2) &\leq& O\left((nL_1)^p(NL_2)^p\right)=O(L^p) \\
                               &=&    O(L^{|\rho_1|/2}).
    \end{eqnarray*}
    When $p=1$, $|\rho_1|=2$, and $|\rho_1|/2=|\rho_1|-1$, so that $Q_{\bar{\pi}}(nL_1,NL_2)\leq O\left(L^{|\rho_1|-1}\right)$, 
    and it is easy to check that we have equality for the only partial permutation in $\text{SPR}_2$, and that $Q_{\bar{\pi}}(nL_1,NL_2)=nNL^{|\rho_1|-1}$ 
    for this $\bar{\pi}$.
    When $p>1$, $|\rho_1|/2<|\rho_1|-1$, so that $Q_{\bar{\pi}}(nL_1,NL_2)=O(L^{|\rho_1|-2})$ for such $\bar{\pi}$. 
  \item The graph of random edges is a tree, and there exist deterministic edges: 
    since any crossidentification between $\{1,...,2p\}$ and $\{2p+1,...,4p\}$ does not give rise to a leaf node when all edges are considered,
    any leafnode in the tree of random edges must be bordering to a deterministic edge.
    Since the tree contains $|\rho_1|+1$ vertices, and since there are at least two leafnodes in any tree, we have that
    $Q_{\bar{\pi}}(nL_1,NL_2)$ has order at most $O\left(L^{|\rho_1|-1}\right)$, 
    with equality only if the graph of random edges borders to exactly two deterministic edges. 
    It is easily seen that this occurs if and only if $|\rho_1|$ pairs of edges are identified in successive order. 
  \item The graph of random edges is not a tree, and there exist deterministic edges: 
    if there are two cycles in the graph of random edges, $Q_{\bar{\pi}}(nL_1,NL_2)$ has order at most $O\left(L^{|\rho_1|-2}\right)$ 
    (two subtracted for the cycles, one for the deterministic edge).
    Similarly, if there is one cycle, and more than one vertex bordering to a deterministic edge, $Q_{\bar{\pi}}(nL_1,NL_2)$ 
    has order at most $O\left(L^{|\rho_1|-2}\right)$. 
    Assume thus that there is only one vertex bordering to a deterministic edge, and only one cycle.
    It is easily checked that this vertex must be on the cycle, and that we must end up in the same situation as in 2) where 
    edges are identified in successive order, for which we actually have a tree. Thus, there is nothing more to consider.
\end{enumerate}
We see that $Q_{\bar{\pi}}(nL_1,NL_2)$ has order at most $O\left(L^{|\rho_1|-1}\right)$ in any case.
Inserting into (\ref{lookatthis}), the first case above contributes with $L^{-1}\frac{1}{nN}$ for $p=1$, 
for $p>1$ we get only terms of order $O(L^{-2})$. 
The third case contributes only with terms of order $O(L^{-2})$. 
For the second case, contributions are of order $O(L^{-2})$ when $|\rho_1|$ pairs of edges are not identified in successive order. 
When they are identified in successive order, we consider the following different possibilities:
\begin{itemize}
  \item When $|\rho_1|$ is odd we will have $k(\rho(\pi))-kd(\rho(\pi))=l(\rho(\pi))-ld(\rho(\pi))=\frac{|\rho_1|-1}{2}$, so that 
    \begin{eqnarray*}
      \lefteqn{Q_{\bar{\pi}}(nL_1,NL_2)}\\
        &=& (L_2N)^{k(\rho(\pi))-kd(\rho(\pi))}(L_1n)^{l(\rho(\pi))-ld(\rho(\pi))}\\ 
        &=& (L_2N)^{\frac{|\rho_1|-1}{2}} (L_1n)^{\frac{|\rho_1|-1}{2}}\\
        &=& (nN)^{\frac{|\rho_1|-1}{2}}L^{\frac{|\rho_1|-1}{2}},
    \end{eqnarray*}
    so that the term for $\bar{\pi}$ in (\ref{lookatthis}) is of order
    $L^{(|\rho_1|-1)/2-|\rho_1|}=L^{-|\rho_1|/2-1/2}$.
    When $|\rho_1|=1$, this is $O\left(L^{-1}\right)$, and the contribution in this case is $\frac{1}{nNL}$
    times the number of partitions in the equivalence class of $\bar{\pi}$
    When $|\rho_1|>1$, all terms are of order $O(L^{-2})$. 
  \item When $|\rho_1|$ is even, either 
    \begin{enumerate}
      \item $k(\rho(\pi))-kd(\rho(\pi))=\frac{|\rho_1|}{2}-1$, $l(\rho(\pi))-ld(\rho(\pi))=\frac{|\rho_1|}{2}$, 
        for which
        \begin{eqnarray*}
          \lefteqn{Q_{\bar{\pi}}(nL_1,NL_2)}\\
            &=& (L_2N)^{\frac{|\rho_1|}{2}-1} (L_1n)^{\frac{|\rho_1|}{2}}\\
            &=& N^{\frac{|\rho_1|}{2}-1} n^{\frac{|\rho_1|}{2}} L^{\frac{|\rho_1|}{2}-1} L_1,
        \end{eqnarray*}
        so that the term for $\bar{\pi}$ in (\ref{lookatthis}) is of order 
        \[ L^{|\rho_1|/2-1-|\rho_1|}L_1=L^{-|\rho_1|/2-1}L_1. \] 
        When the stacking is not vertical, we have that $L_1\leq O(L^{1/2})$, 
        so that the term for $\bar{\pi}$ is of order $\leq O(L^{-|\rho_1|/2-1}L^{1/2})= O(L^{-|\rho_1|/2-1/2})\leq O(L^{-3/2})$ 
        When the stacking is vertical, the term is of order $L^{-|\rho_1|/2}$, which is $O(L^{-2})$ when $|\rho_1|>2$. 
        When $|\rho_1|=2$, the contribution in (\ref{lookatthis}) is seen to be $\frac{1}{N^2L}$ 
        times the number of partitions in the equivalence class of $\bar{\pi}$. 
      \item $k(\rho(\pi))-kd(\rho(\pi))=\frac{|\rho_1|}{2}$, $l(\rho(\pi))-ld(\rho(\pi))=\frac{|\rho_1|}{2}-1$, 
        for which the term for $\bar{\pi}$ in (\ref{lookatthis}) similarly is shown to be of order 
        \[ L^{|\rho_1|/2-1-|\rho_1|}L_2=L^{-|\rho_1|/2-1}L_2, \]
        and, similarly, only horizontal stacking with $|\rho_1|=2$ gives contributions of order $O(L^{-1})$. 
        The contribution in (\ref{lookatthis}) is seen to be $\frac{1}{nNL}$ 
        times the number of partitions in the equivalence class of $\bar{\pi}$.
    \end{enumerate}
\end{itemize}
When it comes to the number of elements in the corresponding equivalence
classes, it is easy to see that
\begin{itemize}
  \item there are $2p^2$ elements for the class where $|\rho_1|=1$,
    corresponding to any choice of the $2p$ edges $\{1,...,2p\}$, and any choice of the $p$ even or odd edges in $\{2p+1,...,4p\}$. 
  \item $p^2$ elements for each class where $|\rho_1|=2$.
\end{itemize} 
Summing up, we see that for $p=1$, $v_{1,\cdot,L} = L^{-1}\frac{2}{nN}D_1+L^{-1}\frac{1}{nN}$ for any type of stacking/averaging. 
For $p\geq 2$ we get that
\begin{eqnarray*} 
  v_{p,R,L} &=& L^{-1} \frac{2p^2}{nN} D_{2p-1} + O\left(L^{-3/2}\right) \\
  v_{p,V,L} &=& L^{-1} \frac{2p^2}{nN} D_{2p-1} + L^{-1} \frac{p^2}{N^2} D_{2p-2} + O\left(L^{-3/2}\right) \\ 
  v_{p,H,L} &=& L^{-1} \frac{2p^2}{nN} D_{2p-1} + L^{-1} \frac{p^2}{nN}  D_{2p-2} + O\left(L^{-3/2}\right),
\end{eqnarray*}
and the first formulas in Theorem~\ref{stackable} follows after multiplying both sides with $L$, and taking limits. 
The case of averaging follows by noting that there are only positive coefficients in the formula (\ref{lookatthis}) for the variance, 
and that the variance is divided by $L$ when one takes $L$ independent observations.  

Finally, we prove why the least variance is obtained when the compound observation matrix is as square as possible. 
With $c_1=nL_1,c_2=NL_2,c=\frac{c_1}{c_2}$, we can write each $Q_{\bar{\pi}}(nL_1,NL_2)$ as a scalar multiple of
\begin{eqnarray*}
  \lefteqn{c_1^kc_2^l+c_1^lc_2^k } \\
  &=& (nN)^{\frac{k+l}{2}}L^{\frac{k+l}{2}}\left( c_1^{\frac{k-l}{2}}c_2^{\frac{l-k}{2}} + c_1^{\frac{l-k}{2}}c_2^{\frac{k-l}{2}}\right) \\
  &=& (nN)^{\frac{k+l}{2}}L^{\frac{k+l}{2}}\left( c^{\frac{k-l}{2}} + c^{\frac{l-k}{2}}\right),
\end{eqnarray*}
It is clear that $f(c)=c^{(k-l)/2} + c^{(l-k)/2}$ has a global minimum at $c=1$ on $(0,\infty)$, and the result follows. 

\bibliography{../bib/mybib,../bib/mainbib}

\begin{thebibliography}{10}
\providecommand{\url}[1]{#1}
\csname url@samestyle\endcsname
\providecommand{\newblock}{\relax}
\providecommand{\bibinfo}[2]{#2}
\providecommand{\BIBentrySTDinterwordspacing}{\spaceskip=0pt\relax}
\providecommand{\BIBentryALTinterwordstretchfactor}{4}
\providecommand{\BIBentryALTinterwordspacing}{\spaceskip=\fontdimen2\font plus
\BIBentryALTinterwordstretchfactor\fontdimen3\font minus
  \fontdimen4\font\relax}
\providecommand{\BIBforeignlanguage}[2]{{%
\expandafter\ifx\csname l@#1\endcsname\relax
\typeout{** WARNING: IEEEtran.bst: No hyphenation pattern has been}%
\typeout{** loaded for the language `#1'. Using the pattern for}%
\typeout{** the default language instead.}%
\else
\language=\csname l@#1\endcsname
\fi
#2}}
\providecommand{\BIBdecl}{\relax}
\BIBdecl

\bibitem{paper:telatar99}
E.~Telatar, ``Capacity of multi-antenna gaussian channels,'' \emph{Eur. Trans.
  Telecomm. ETT}, vol.~10, no.~6, pp. 585--596, Nov. 1999.

\bibitem{book:bouchaud}
J.-P. Bouchaud and M.~Potters, \emph{Theory of Financial Risk and Derivative
  Pricing - From Statistical Physics to Risk Management}.\hskip 1em plus 0.5em
  minus 0.4em\relax Cambridge: Cambridge University Press, 2000.

\bibitem{paper:guhr}
T.~Guhr, A.~M\"uller-Groeling, and H.~A. Weidenm\"uller, ``Random matrix
  theories in quantum physics: Common concepts,'' \emph{Phys.Rept. 299}, pp.
  189--425, 1998.

\bibitem{vo2}
D.~V. Voiculescu, ``Addition of certain non-commuting random variables,''
  \emph{J. Funct. Anal.}, vol.~66, pp. 323--335, 1986.

\bibitem{paper:vomult}
------, ``Multiplication of certain noncommuting random variables,'' \emph{J.
  Operator Theory}, vol.~18, no.~2, pp. 223--235, 1987.

\bibitem{vo6}
D.~Voiculescu, ``Circular and semicircular systems and free product factors,''
  \emph{Operator algebras, unitary representations, enveloping algebras and
  invariant theory}, vol.~92, 1990.

\bibitem{vo7}
------, ``Limit laws for random matrices and free products,'' \emph{Inv.
  Math.}, vol. 104, pp. 201--220, 1991.

\bibitem{book:hiaipetz}
F.~Hiai and D.~Petz, \emph{The Semicircle Law, Free Random Variables and
  Entropy}.\hskip 1em plus 0.5em minus 0.4em\relax American Mathematical
  Society, 2000.

\bibitem{ryandebbah:vandermonde1}
{\O}.~Ryan and M.~Debbah, ``Asymptotic behaviour of random {V}andermonde
  matrices with entries on the unit circle,'' \emph{IEEE Trans. on Information
  Theory}, vol.~55, no.~7, pp. 3115--3148, 2009.

\bibitem{ryandebbah:vandermonde2}
------, ``Convolution operations arising from {V}andermonde matrices,''
  \emph{Submitted to IEEE Trans. on Information Theory}, 2009.

\bibitem{eurecom:freedeconvinftheory}
------, ``Free deconvolution for signal processing applications,''
  \emph{Submitted to IEEE Trans. on Information Theory}, 2007,
  http://arxiv.org/abs/cs.IT/0701025.

\bibitem{ryandebbah:finitedim}
{\O}.~Ryan, A.~Masucci, S.~Yang, and M.~Debbah, ``Finite dimensional
  statistical inference,'' \emph{Submitted to IEEE Trans. on Information
  Theory}, 2009.

\bibitem{paper:doziersilverstein1}
B.~Dozier and J.~W. Silverstein, ``On the empirical distribution of eigenvalues
  of large dimensional information-plus-noise type matrices,'' \emph{J.
  Multivariate Anal.}, vol.~98, no.~4, pp. 678--694, 2007.

\bibitem{Florent}
F.~Benaych-Georges and M.~Debbah, ``Free deconvolution: from theory to
  practice,'' \emph{submitted to IEEE Transactions on Information Theory},
  2008.

\bibitem{romain:sensing}
R.~Couillet, {\O}.~Ryan, and M.~Debbah, ``A free probability approach to
  collaborative multi-cell sensing,'' \emph{In preparation}, 2009.

\bibitem{eurecom:channelcapacity}
{\O}.~Ryan and M.~Debbah, ``Channel capacity estimation using free probability
  theory,'' \emph{IEEE Trans. Signal Process.}, vol.~56, no.~11, pp.
  5654--5667, November 2008.

\bibitem{secondorderfreeness1}
J.~A. Mingo and R.~Speicher, ``Second order freeness and fluctuations of random
  matrices: {I}. {G}aussian and {W}ishart matrices and cyclic {F}ock spaces,''
  \emph{J. Funct. Anal.}, vol. 235, no.~1, pp. 226--270, 2006.

\bibitem{secondorderfreeness2}
J.~A. Mingo, P.~\'{S}niady, and R.~Speicher, ``Second order freeness and
  fluctuations of random matrices: {II}. unitary random matrices,'' \emph{Adv.
  in Math.}, vol. 209, pp. 212--240, 2007.

\bibitem{secondorderfreeness3}
B.~Collins, J.~A. Mingo, P.~\'{S}niady, and R.~Speicher, ``Second order
  freeness and fluctuations of random matrices: {III}. higher order freeness
  and free cumulants,'' \emph{Documenta Math.}, vol.~12, pp. 1--70, 2007.

\bibitem{ryandebbah:optstackingtools}
{\O}.~Ryan, \emph{Tools for the optimal stcaking of noisy observations}, 2010,
  http://folk.uio.no/oyvindry/findimstacking/.

\end{thebibliography}

\end{document}